%% file: ft-hd.tex
\documentclass[11pt]{article}
\usepackage{fullpage}
\usepackage{times}
\usepackage{subfigure}
\usepackage{color}
\usepackage{url}
\usepackage{graphicx}
\usepackage{amsmath}
\usepackage{amsthm}
\usepackage{amssymb}
\usepackage{algorithm}
\usepackage{algpseudocode}
\usepackage{mathtools}

\usepackage{tikz}
\usetikzlibrary{positioning, calc, chains}

{\makeatletter
 \gdef\xxxmark{%
   \expandafter\ifx\csname @mpargs\endcsname\relax 
     \expandafter\ifx\csname @captype\endcsname\relax 
       \marginpar{xxx}
     \else
       xxx 
     \fi
   \else
     xxx 
   \fi}
 \gdef\xxx{\@ifnextchar[\xxx@lab\xxx@nolab}
 \long\gdef\xxx@lab[#1]#2{{\bf [\xxxmark #2 ---{\sc #1}]}}
 \long\gdef\xxx@nolab#1{{\bf [\xxxmark #1]}}
}

\newcommand{\abs}[1]{|#1|}

\newcommand{\norm}[2]{\lVert#2\rVert_{#1}}

\newcommand{\wh}{\widehat}

\newcommand{\twid}[1]{\tilde{#1}}

\DeclareMathOperator{\supp}{supp}

\def\Z{\mathbb{Z}}
\def\C{\mathbb{C}}
\def\F{\mathcal{F}}
\def\B{\mathbb{B}}

\newcommand{\expect}{{\bf \mbox{\bf E}}}
\newcommand{\prob}{{\bf \mbox{\bf Pr}}}

\definecolor{gray}{rgb}{0.5,0.5,0.5}
\newcommand{\e}{{\epsilon}}

\newtheorem{theorem}{Theorem}[section]
\newtheorem{lemma}[theorem]{Lemma}

\newtheorem{definition}[theorem]{Definition}
\newtheorem{remark}[theorem]{Remark}

\newenvironment{proofof}[1]{\noindent{\bf Proof of #1:}}{$\qed$\par}
\usepackage{hyperref}
\hypersetup{
bookmarksnumbered
}

\DeclareMathOperator{\err}{Err}
\newcommand{\poly}{\text{poly}}

\newcommand{\fc}{F}
\newcommand{\gl}{\mathcal{M}_{d\times d}}

\begin{document}

\begin{titlepage}
\title{Sample-Optimal Fourier Sampling in Any Constant Dimension -- Part I}
\author{Piotr Indyk \and Michael Kapralov}

\maketitle

\begin{abstract}
We give an algorithm for $\ell_2/\ell_2$ sparse recovery from Fourier measurements using $O(k\log N)$ samples, matching the lower bound of \cite{DIPW} for non-adaptive algorithms up to constant factors for any $k\leq N^{1-\delta}$. The algorithm runs in $\tilde O(N)$ time. Our algorithm extends to higher dimensions, leading to sample complexity of $O_d(k\log N)$, which is optimal up to constant factors for any $d=O(1)$.  These are the first sample optimal algorithms for these problems. 

A preliminary experimental evaluation indicates that our algorithm has empirical sampling complexity comparable to that of other recovery methods known in the literature, while  providing strong provable guarantees on the recovery quality.
\end{abstract}
\thispagestyle{empty}
\end{titlepage}

\newcommand{\nsq}{[n]^d}

\input{intro.tex}

\input{preliminaries.tex}

\input{filter.tex}
\input{linear.tex}
\input{experiments.tex}

\newcommand{\etalchar}[1]{$^{#1}$}

\begin{appendix}
\input{app.tex}

\end{appendix}
\end{document}

%% file: intro.tex

\section{Introduction}

The Discrete Fourier Transform (DFT) is a mathematical notion that allows to represent a sampled signal or function as a combination of  discrete frequencies. It is a powerful tool used in many areas of science and engineering. Its popularity stems  from the fact that signals are typically easier to  process and interpret when represented in the frequency domain. As a result, DFT plays a key role in  digital signal processing, image processing, communications,  partial differential equation solvers, etc.  Many of these applications rely on the fact that most of the Fourier coefficients of the signals are small or equal to zero, i.e., the signals are (approximately) {\em sparse}. For example, sparsity provides the rationale underlying compression schemes for audio, image and video signals, since keeping the top few coefficients often suffices to preserve most of the signal energy. 

An attractive property of sparse signals is that they can be acquired from only a small number of samples. Reducing the sample complexity is highly desirable as it implies a reduction in signal acquisition time, measurement overhead
and/or communication cost. For example, one of the main goals in medical imaging is to reduce the sample complexity in order to reduce the time the patient spends in the MRI machine~\cite{MRICS}, or the radiation dose received~\cite{CT_CS}.  Similarly in spectrum sensing, a lower average sampling rate enables the fabrication of efficient analog to digital converters (ADCs) that can acquire very wideband multi-GHz signals~\cite{CS_ADC}. As a result, designing sampling schemes and the associated sparse recovery algorithms has been a subject of extensive research in multiple areas, such as:
\begin{itemize}
\item {\em Compressive sensing:}  The area of compressive sensing~\cite{Don,CTao}, developed over the last decade, studies the task of recovering (approximately) sparse signals from linear measurements. Although several classes of linear measurements were studied, acquisition of sparse signals using few Fourier measurements (or, equivalently, acquisition of Fourier-sparse signals using few signal samples) has been one of the key problems studied in this area. In particular, the seminal work of~\cite{CTao,RV} has shown that one can recover $N$-dimensional signals with at most $k$ Fourier coefficients using only $k \log^{O(1)} N$ samples. The recovery algorithms are based on linear programming and run in time polynomial in $N$. See~\cite{FR} for an introduction to the area.
\item {\em Sparse Fourier Transform:} A different line of research, with origins in computational complexity and learning theory, has been focused on  developing algorithms whose sample complexity {\em and} running time bounds scale with the sparsity. Many such algorithms have been proposed in the literature, including \cite{GL,KM,Man,GGIMS,AGS,GMS,Iw,Ak,HIKP,HIKP2,LWC,BCGLS,HAKI,pawar2013computing,heidersparse, IKP}. These works show that, for a wide range of signals, both the time complexity and the number of signal samples taken can be significantly sub-linear in $N$. 
\end{itemize}

The best known results obtained in both of those areas are summarized in the following table. 
For the sake of uniformity
we focus on algorithms that work for general signals and recover $k$-sparse approximations satisfying the so-called $\ell_2/\ell_2$ approximation
guarantee\footnote{Some of the algorithms~\cite{CTao,RV,CGV} can in
  fact be made deterministic, but at the cost of satisfying a somewhat
  weaker $\ell_2/\ell_1$ guarantee. Also, additional results that hold
  for exactly sparse signals are known, see e.g.,~\cite{BCGLS} and
  references therein.}.
In this case, the goal of an algorithm is as follows: given $m$ samples of the Fourier transform $\wh{x}$ of a signal $x$\footnote{Here and for the rest of this paper, we will consider the \emph{inverse} discrete Fourier transform problem of estimating a sparse $x$ from
samples of $\wh{x}$.  This leads to a simpler notation.  Note that the the forward and inverse DFTs are equivalent modulo conjugation.},  and the sparsity parameter $k$, output $x'$ satisfying
\begin{equation}
\label{e:l2l2}
\| x-x'\|_2 \le C \min_{k \text{-sparse } y }  \|x-y\|_2,
\end{equation}
The algorithms are randomized and succeed with constant probability.

\begin{figure}[h!]
\begin{center}
\begin{tabular}{|c|c|c|c|c|}
\hline
Reference & Time & Samples & Approximation & Signal model\\
\hline\hline
\cite{CTao,RV} & & & & \\
 \cite{CGV}  & $N \times m$ linear program & $O(k \log^3(k) \log(N))$ & $C=O(1)$ & worst case\\
\cite{CP} & $N \times m$ linear program & $O(k \log N)$ & $C=(\log N)^{O(1)}$ & worst case\\
\cite{HIKP2} & $O(k \log(N) \log(N/k))$ & $O(k \log(N) \log(N/k))$   &   any $C>1$ & worst case \\
\cite{GHIKPS} & $O(k \log^2 N)$ & $O(k \log N)$ & $C=O(1)$ & average case, \\
& & & & $k=\Theta(\sqrt{N})$\\
\cite{pawar2014} & $O(N \log N)$ & $O(k \log N)$ & $C=O(1)$ & average case,\\
& & & & $k=O(N^{\alpha})$, $\alpha<1$\\
\cite{IKP} &  $O(k \log^2(N)  \log^{O(1)} \log N)$ & $O(k \log(N) \log^{O(1)} \log N)$ & any $C>1$ & worst case\\
\hline
\cite{DIPW} & & $\Omega(k \log(N/k))$ & constant $C$ & lower bound\\
\hline
\end{tabular}
\end{center}
\caption{Bounds for the algorithms that recover $k$-sparse Fourier approximations . All algorithms produce an output satisfying Equation~\ref{e:l2l2} with probability of success that is at least constant.}
\end{figure}

As evident from the table, none of the results obtained so far was able to guarantee sparse recovery from the optimal number of samples, unless either the approximation factor was super-constant or the result held for average-case signals.  In fact, it was not  even known whether there is an {\em exponential time} algorithm that uses only $O(k \log N)$ samples in the worst case. 

 A second limitation, that applied to the sub-linear time algorithms in the last three rows in the table, but not to compressive sensing algorithms in the first two rows of the table,  is that those algorithms were  designed for {\em one-dimensional} signals. However, the  sparsest signals often occur in applications involving higher-dimensional DFTs, since they involve much larger signal lengths $N$. Although one can reduce, e.g.,  the two-dimensional DFT over $p \times q$ grid  to the
 one-dimensional DFT over a signal of length $pq$~\cite{GMS,Iw-arxiv}),  the reduction applies only if $p$
 and $q$ are relatively prime.
 This excludes the most typical case of
 $m \times m$ grids where $m$ is a power of $2$. The only prior
 algorithm that applies to general $m\times m$ grids, due
 to~\cite{GMS}, has $O(k \log^c N)$ sample and time complexity for a
 rather large value of $c$.  If $N$ is a power of $2$, a
 two-dimensional adaptation of the~\cite{HIKP2} algorithm (outlined in~\cite{GHIKPS}) 
has roughly $O(k \log^3 N)$ time and sample complexity, and an adaptation of~\cite{IKP} has $O(k\log^2 N(\log\log N)^{O(1)})$ sample complexity.

\paragraph{Our results} In this paper we give an algorithm that overcomes both of the aforementioned limitations. Specifically, we present an algorithm for the  sparse Fourier transform in any fixed dimension that uses only $O(k \log N)$ samples of the signal. This is the first algorithm that matches the lower bound of $\cite{DIPW}$, for $k$ up to $N^{1-\delta}$ for any constant $\delta>0$. The recovery algorithm runs in time $O(N \log^{O(1)} N)$.

In addition, we note that the algorithm is in fact quite simple. It is essentially a variant of an iterative thresholding scheme, where the coordinates of the signal are updated sequentially in order to minimize the difference between the current approximation and the underlying signal. In Section~\ref{sec:experiments} we discuss a preliminary experimental evaluation of this algorithm, which shows promising results. 

The techniques introduced in this paper have already found applications. In particular, in a followup paper~\cite{IK14b}, we give an algorithm that uses $O(k \log(N) \log^{O(1)} \log N)$ samples of the signal and has the running time of $O(k \log^{O(1)}(N)  \log^{O(1)} \log N)$ for any constant $d$. This generalizes the result of~\cite{IKP} to any constant dimension, at the expense of somewhat larger runtime.

\paragraph{Our techniques}  The overall outline of our algorithms follows the framework of~\cite{GMS, HIKP2,IKP}, which adapt the methods of~\cite{CCF,GLPS} from arbitrary linear
measurements to Fourier ones.  The idea is to take, multiple times,  a set of $B=O(k)$
linear measurements of the form
\[
\twid{u}_j = \sum_{i : h(i) = j} s_i x_i 
\]
for random hash functions $h: [N] \to [B]$ and random sign changes
$s_i$ with $\abs{s_i} = 1$.  This denotes \emph{hashing} to $B$
\emph{buckets}.  With such ideal linear measurements, $O(\log(N/k))$
hashes suffice for sparse recovery, giving an $O(k \log(N/k))$ sample
complexity.

The sparse Fourier transform algorithms approximate $\twid{u}$ using linear combinations of Fourier samples. Specifically, the coefficients of $x$ are first pseudo-randomly permuted, by re-arranging the access to $\hat{x}$  via a random affine permutation.  Then the coefficients are partitioned into buckets. This steps uses  the``filtering'' process that approximately partitions the range of $x$ into intervals (or, in higher dimension, squares) with $N/B$ coefficients each, and collapses each interval into one bucket. 
To minimize the number of samples taken, the filtering process is approximate. In particular the coefficients  contribute (``leak''') to buckets other than the one they are nominally mapped into, although that contribution is limited and controlled by the quality of the filter. The details are described in Section~\ref{sec:filter}, see also~\cite{HIKP} for further overview.


Overall, this probabilistic process ensures that
 most of the large coefficients are ``isolated'', i.e.,  are hashed to unique buckets, as well as that the contributions from the  ``tail'' of the signal $x$ to those buckets is not much greater than the average; the tail of the signal is defined as $\err_k(x)=\min_{k-\text{sparse}~y} ||x-y||_2$.
 This enables the algorithm to identify the positions of the large coefficients, as well as estimate their values, producing a sparse estimate $\chi$ of $x$.  To improve this
estimate, we repeat the process on $x - \chi$ by subtracting the
influence of $\chi$ during hashing.  The repetition will yield a good
sparse approximation $\chi$ of $x$.

To achieve the optimal number of measurements,  however, our algorithm departs from the above scheme in a crucial way: the algorithm does {\em not} use fresh hash functions in every repetition. Instead, $O(\log N)$ hash functions are chosen at the beginning of the process, such that each large coefficient is isolated by most of those functions with high probability. 
 The same hash functions are then used throughout the duration of the algorithm. Note that  each hash function requires a separate set of samples to construct the buckets, so reusing the hash functions means that the number of samples does not grow with the number of iterations. This enables us to achieve the optimal measurement bound. 
 
 At the same time reusing the hash functions creates a major difficulty: if the algorithm identifies a non-existing large coefficient by mistake and adds it to $\chi$, this coefficient will be present in the difference vector $x - \chi$ and will need to be corrected later. And unlike the earlier guarantee for the large coefficients of the {\em original} signal $x$, we do not have any guarantees that large {\em erroneous} coefficients will be isolated by the hash functions, since the positions of those coefficients are determined by those functions. Because of these difficulties, almost all prior works\footnote{We are only aware of two exceptions: the algorithms of~\cite{GHIKPS,pawar2013computing} (which were analyzed only for the easier case where the large coefficients  themselves were randomly distributed) and the analysis of iterative thresholding schemes due to~\cite{BLM} (which relied on the fact that the measurements were induced by Gaussian or Gaussian-like matrices).}  either used a fresh set of measurements in each iteration (almost all sparse Fourier transform algorithms fall into this category) or provided stronger deterministic guarantees for the sampling pattern (such as  the restricted isometry property~\cite{CTao}).
 However, the latter option required a larger number of measurements to ensure the desired properties.
 Our algorithm circumvents this difficulty by ensuring that no large coefficients are created erroneously. This is non-trivial, since the hashing process is quite noisy (e.g, the bucketing process suffers from leakage). Our solution is to recover the large coefficients in the decreasing order of their magnitude. Specifically, in each step, we recover coefficients with magnitude that exceeds a specific threshold (that decreases exponentially).
 The process is designed to ensure  that (i) all coefficients above the threshold are recovered and (ii) all recovered coefficients have magnitudes close to the threshold. In this way the set of locations of large coefficients stays fixed (or monotonically decreases) over the duration of the algorithms, and we can ensure the isolation properties of those coefficients during the initial choice of the hash functions. 

Overall, our algorithm has two key properties (i) it is iterative, and therefore the values of the coefficients estimated in one stage can be corrected in the second stage and (ii) does not require fresh hash functions (and therefore new measurements) in each iteration. Property (ii) implies that the number of measurements is determined by only a single (first) stage, and does not increase beyond that. Property (i) implies that the bucketing and estimation process can be achieved using rather ``crude'' filters\footnote{In fact, our filters are only slightly more accurate than the filters introduced in~\cite{GMS}, and require the same number of samples.}, since the estimated values can be corrected in the future stages. As a result each of the hash function require only $O(k)$ samples;  since we use $O(\log N)$ hash functions, the $O(k \log N)$ bound follows. This stands in contrast with the algorithm of $\cite{GMS}$ (which used crude filters of similar complexity but required new measurements per each iteration) or~\cite{HIKP2} (which used much stronger filters with $O(k \log N)$ sample complexity) or~\cite{IKP} (which used filters of varying quality and sample complexity). The advantage of our approach is amplified in higher dimension, as the ratio of the number of samples required by the filter to the value $k$ grows exponentially in the dimension. Thus, our filters still require $O(k)$ samples in any fixed dimension $d$, while for~\cite{HIKP2, IKP} this bound increases to $O(k \log^d N)$. 

\paragraph{Organization} We give definitions and basic results relevant to sparse recovery from Fourier measurements in section~\ref{sec:prelim}. Filters that our algorithm uses are constructed in section~\ref{sec:filter}. Section~\ref{sec:algo-linear} states the algorithm and provides intuition behind the analysis. The main lemmas of the analysis are proved in section~\ref{sec:isolated}, and full analysis of the algorithm is provided in section~\ref{sec:analysis}. Results of an experimental evaluation are presented in section~\ref{sec:experiments}, and ommitted proofs are given in Appendix~\ref{app:A}.

%% file: preliminaries.tex
\section{Preliminaries}\label{sec:prelim}

For a positive even integer $a$ we will use the notation $[a]=\{-\frac{a}{2}, -\frac{a}{2}+1, \ldots, -1, 0, 1,\ldots, \frac{a}{2}-1\}$. We will consider signals of length $N=n^d$, where $n$ is a power of $2$ and $d\geq 1$ is the dimension.  We use the notation $\omega=e^{2\pi i/n}$ for the root of unity of order $n$. The $d$-dimensional forward and inverse Fourier transforms are given by
\begin{equation}\label{eq:dft-forward}
\hat x_{j}=\frac1{\sqrt{N}}\sum_{i\in \nsq}  \omega^{-i^Tj}x_i \text{~~and~~}x_{j}=\frac1{\sqrt{N}}\sum_{i\in \nsq}  \omega^{i^Tj}\hat x_i
\end{equation}
respectively, where $j\in \nsq$. We will denote the forward Fourier transform by $\F$ and 
Note that we use the orthonormal version of the Fourier transform. Thus, we have $||\hat x||_2=||x||_2$ for all $x\in \C^N$ (Parseval's identity).
We recover a signal $z$ such that 
\begin{equation*}
||x-z||_2\leq (1+\e)\min_{k-\text{~sparse~} y} ||x-y||_2
\end{equation*}
from samples of $\wh{x}$.

We will use pseudorandom spectrum permutations, which we now define.  We write $\gl$ for the set of $d\times d$ matrices  over $\Z_n$ with odd determinant.
For $\Sigma\in \gl, q\in \nsq$ and $i\in \nsq$ let 
$\pi_{\Sigma, q}(i)=\Sigma(i-q) \mod n$.
Since $\Sigma\in \gl$, this is a permutation. Our algorithm will use $\pi$ to hash heavy hitters into $B$ buckets, where we will choose $B\approx k/\e$. It should be noted that unlike many sublinear time algorithms for the problem, our algorithm does not use $O(k)$ buckets with centers equispaced in the time domain. Instead, we think of each point in time domain as having a bucket around it. This imporoves the dependence of the number of samples on the dimension $d$. We will often omit the subscript $\Sigma, q$ and simply write $\pi(i)$ when $\Sigma, q$ is fixed or clear from context. For $i, j\in \nsq$ we let $o_i(j)= \pi(j) - \pi(i)$ to be the ``offset'' of $j\in \nsq$ relative to $i\in \nsq$. We will always have $B=b^d$, where $b$ is a power of $2$. 
\begin{definition}
  Suppose that $\Sigma^{-1}$ exists $\bmod~n$. For $a, q\in \nsq$ we define the
  permutation $P_{\Sigma, a, q}$ by $(P_{\Sigma, a, q}\hat x)_i=\hat x_{\Sigma^T(i-a)} \omega^{i^T\Sigma q}$.  
\end{definition}
\begin{lemma}\label{lm:perm}
$\F^{-1}({P_{\Sigma, a, q} \hat x})_{\pi_{\Sigma, q}(i)}=x_i \omega^{a^T \Sigma i}$
\end{lemma}
The proof is similar to the proof of Claim~B.3 in \cite{GHIKPS} and is given in Appendix~\ref{app:A} for completeness. 
Define 
\begin{equation}\label{eq:mu-def}
\begin{split}
\err_k(x)=\min_{k-\text{sparse}~y} ||x-y||_2\text{~~and~~}\mu^2=\err_k^2(x)/k.
\end{split}
\end{equation}
In this paper, we assume knowledge of $\mu$ (a constant factor upper bound on $\mu$ suffices). We also assume that the signal to noise ration is bounded by a polynomial, namely 
that $R^*:=||x||_\infty/\mu\leq n^C$ for a constant $C>0$.
We use the notation $\B^\infty_{r}(x)$ to denote the $\ell_\infty$ ball of radius $r$ around $x$:
$$
\B^\infty_{r}(x)=\{y\in \nsq: ||x-y||_\infty\leq r\},
$$
where $||x-y||_\infty=\max_{s\in d} ||x_s-y_s||_{\circ}$, and $||x_s-y_s||_{\circ}$ is the circular distance on $\mathbb{Z}_n$. 
We will also use the notation $f\lesssim g$ to denote $f=O(g)$.

%% file: filter.tex
\section{Filter construction and properties}\label{sec:filter}
For an integer  $b>0$ a power of $2$ let
\begin{equation}\label{eq:boxcar-time}
\hat H_i^1=\left\{
\begin{array}{ll}
\frac{\sqrt{n}}{b-1},&\text{~if~}|i|<b/2\\
0&\text{o.w.}
\end{array}
\right.
\end{equation}

Let $\hat H^{\fc}$ denote the $\fc$-fold convolution of $\hat H^1$ with itself, so that $\text{supp~} \hat H^\fc\subseteq [-\fc \cdot b, \fc\cdot b]$. Here and below $\fc$ is a parameter that we will choose to satisfy $\fc\geq 2d, \fc=\Theta(d)$. The Fourier transform of $\hat H^1$ is the Dirichlet kernel (see e.g.~\cite{book:fourier}, page 37):
\begin{equation*}
\begin{split}
H^{1}_j&=\frac1{b-1}\sum_{|i|< b/2} \omega^{ij}=\frac{\sin(\pi (b-1) j/n)}{(b-1)\sin (\pi j/n)}\text{~for~}j\neq 0\\
H^{1}_0&=1.\\
\end{split}
\end{equation*}
Thus,
$H^{\fc}_j=\left(\frac1{b-1}\sum_{|i|< b/2} \omega^{ij}\right)^{\fc}=\left(\frac{\sin(\pi (b-1) j/n)}{(b-1)\sin (\pi j/n)}\right)^{\fc}$ for $j\neq 0$,
and $H^{\fc}_0=1$. For $i\in \nsq$ let 
\begin{equation}\label{eq:def-g}
G_i=\prod_{s=1}^d H^{\fc}_{i_s},
\end{equation}
so that $\hat G_i=\prod_{s=1}^d \hat H^{\fc}_{i_s}$ and $\supp \hat G\subseteq [-F\cdot b, F\cdot b]^d$. 
 We will use the following simple properties of $G$:
\begin{lemma}\label{lm:filter-prop}
For any  $\fc\geq 1$ one has
\begin{description}
\item[1] $G_0=1$, and $G_j\in [\frac1{(2\pi)^{\fc \cdot d}}, 1]$ for all $j\in \nsq$ such that $||j||_\infty\leq \frac{n}{2b}$;
\item[2] $|G_j|\leq \left(\frac2{1+(b/n)||j||_\infty}\right)^{\fc}$ for all $j\in \nsq$
\end{description}
as long as $b\geq 3$.
\end{lemma}
The two properties imply that most of the mass of the filter is concentrated in a square of side $O(n/b)$, approximating the ``ideal'' filter (whose value would be equal to $1$ for entries within the square and equal to $0$ outside of it). The proof of the lemma is similar to the analysis of filters in \cite{HIKP, IKP} and is given in Appendix~\ref{app:A}. 
We will not use the lower bound on $G$ given in the first claim of Lemma~\ref{lm:filter-prop} for our $\tilde O(N)$ time algorithm in this paper. We state the Lemma in full form for later use in~\cite{IK14b}, where we present a sublinear time algorithm.
 
The following property of pseudorandom permutations $\pi_{\Sigma, q}$ makes hashing using our filters effective (i.e. allows us to bound noise in each bucket, see Lemma~\ref{lm:hashing}, see below):
\begin{lemma}\label{lemma:limitedindependence}
Let $i, j\in \nsq$. Let $\Sigma$ be uniformly random with odd determinant. Then for all $t\geq 0$
  \[
\prob[||\Sigma(i-j)||_\infty \leq t] \leq 2(2t/n)^d.
  \]
\end{lemma}
A somewhat incomplete proof of this lemma for the case $d=2$ appeared as Lemma B.4 in  \cite{GHIKPS}. We give a full proof for arbitrary $d$ in Appendix~\ref{app:A}.

We access the signal $x$ via random samples of $\hat x$, namely by computing the signal $\F^{-1}((P_{\Sigma, a, q}\hat x)\cdot \hat G)$. As Lemma~\ref{lm:hashing} below shows, this effectively ``hashes''  $x$ into $B=b^d$ bins by convolving it with the filter $G$ constructed above. Since our algorithm runs in $\tilde O(N)$ as opposed to $\tilde O(k)$ time, we can afford to work with bins around any location in time domain (we will be interested in locations of heavy hitters after applying the permutation, see Lemma~\ref{lm:hashing}).  This improves the dependence of our sample complexity on $d$. The properties of the filtering process are summarized in 
\begin{lemma}\label{lm:hashing}
Let $x\in \C^N$. Choose $\Sigma\in \gl, a, q\in \nsq$  uniformly at random, independent of $x$.   Let 
$$
u =\sqrt{N}\F^{-1}((P_{\Sigma, a, q}\hat x)\cdot \hat G),
$$
where $G$ is the filter constructed in ~\eqref{eq:def-g}. Let $\pi=\pi_{\Sigma, q}$.

For  $i\in \nsq$ let
$\mu^2_{\Sigma, q}(i)=\sum_{j \in \nsq\setminus \{i\}} \abs{x_j G_{o_i(j)}}^2$, where $o_i(j)= \pi(j) - \pi(i)$ as before.
Suppose that $\fc\geq 2d$. Then for any $i\in \nsq$
\begin{enumerate}
  \item $\expect_{\Sigma, q}[\mu^2_{\Sigma, q}(i)] \leq  C^d \norm{2}{x}^2/B$ for a constant $C>0$.
  \item for any $\Sigma, q$ one has 
  $  \expect_{a}[\abs{\omega^{-a^T\Sigma i}u_{\pi(i)} - x_i}^2]\lesssim \mu^2_{\Sigma, q}(i)+\delta ||x||_2^2$,
  where the last term corresponds to the numerical error incurred from computing FFT with $O(\log 1/\delta)$ bits of machine precision.
\end{enumerate}
\end{lemma}
The proof of Lemma~\ref{lm:hashing} is given in Appendix~\ref{app:A}.
\begin{remark}
We assume throughout the paper that arithmetic operations are performed on $C\log N$ bit numbers for a sufficiently large constant $C>0$ such that $\delta ||x||_2^2\leq \delta (R^*)^2n\mu^2\leq \mu^2/N$, so that the effect of rounding errors on Lemma~\ref{lm:hashing} is negligible.
\end{remark}

%% file: linear.tex
\section{The algorithm}\label{sec:algo-linear}
In this section we present our $\tilde O(N)$ time algorithm that achieves $d^{O(d)}\frac1{\e}k\log N$ sample complexity and give the main definitions required for its analysis.
Our algorithm follows the natural iterative recovery scheme. The main body of the algorithm (Algorithm~\ref{alg:sfft}) takes samples of the signal $\wh{x}$ and repeatedly calls the \textsc{LocateAndEstimate} function (Algorithm~\ref{alg:loc-est}), improving estimates of the values of dominant elements of $x$ over $O(\log n)$ iterations. Crucially, samples of $\wh{x}$ are only taken at the beginning of Algorithm~\ref{alg:sfft} and passed to each invocation of \textsc{LocateAndEstimate}. Each invocation of \textsc{LocateAndEstimate} takes samples of $\wh{x}$ as well as the current approximation $\chi$ to $x$ as input, and outputs a constant factor approximation to dominant elements of $x-\chi$ (see section~\ref{sec:analysis} for analysis of \textsc{LocateAndEstimate}).

\begin{algorithm}[H]
\caption{Overall algorithm: perform Sparse Fourier Transform}
\label{alg:sfft} 
\begin{algorithmic}[1] 
\Procedure{SparseFFT}{$\hat x, k, \e, R^*, \mu$}\Comment{$R^*$ is a bound on $||x||_\infty/(\sqrt{\e}\mu)$}
\State $\chi^{(0)} \gets 0$ \Comment{in $\C^n$}.\Comment{$\mu$ is the noise level (defined in \eqref{eq:mu-def})}
\State $T \gets \log_2 R^*$
\State $B\leftarrow k/(\e\alpha^{d})$ \Comment{Choose $\alpha$ so that $B=b^d$ for $b$ a power of $2$}
\State $G, \wh{G}\gets$ filter as in \eqref{eq:def-g}
\State $r_{max}\gets \Theta(\log N)$
\For{$r=0$ to $r_{max}$}
\State Choose $\Sigma_r\in \gl, a_r, q_r\in \nsq$ uniformly at random
\State For $r=1,\ldots, r_{max}$, $u^{r}\gets \sqrt{N}\F^{-1}((P_{\Sigma, a, q}\hat x)\cdot \hat G)$
\State $\rhd${~Note that $u^r\in \C^{\nsq}$ for all $r$}
\EndFor
\For{$t = 0, 1, \dotsc,T-1$}
\State $\chi' \gets \Call{LocateAndEstimate}{\hat x,  \chi^{(t)}, \{(\Sigma_r, a_r, b_r), u_r\}_{r=1}^{r_{max}}, r_{max}, \wh{G}, 4\sqrt{\e}\mu 2^{T-(t+1)}}$
\State $\chi^{(t+1)} \gets \chi^{(t)} + \chi'$
\EndFor
\State \textbf{return} $\chi^{(T)}$
\EndProcedure 
\end{algorithmic}
\end{algorithm}

\begin{algorithm}
\caption{\textsc{LocateAndEstimate}($\hat x, \chi, \{(\Sigma_r, a_r, q_r), u_r\}_{r=1}^{r_{max}}, r_{max}, \wh{G}, \nu$)}\label{alg:loc-est} 
\begin{algorithmic}[1] 
\Procedure{LocateAndEstimate}{$\hat x, \chi, \{(\Sigma_r, a_r, q_r), u_r\}_{r=1}^{r_{max}}, r_{max}, \wh{G}, \nu$}
\State {\bf Requires} that $||x-\chi||_\infty\leq 2\nu$
\State {\bf Guarantees} that $||x-\chi-\chi'||_\infty\leq \nu$
\State $L\gets \emptyset$
\State $w\gets 0$
\For {$r=0$ to $r_{max}$}
\State $v^{r}\gets u^r - \sqrt{N}\F^{-1}((P_{\Sigma, a, q}\hat \chi)\cdot \hat G)$\Comment{Update signal: note this does not use any new samples}
\EndFor
\For{$f\in \nsq$}
\State $S\gets \emptyset$
\For {$r=0$ to $r_{max}$}
\State Denote permutation $\pi_{\Sigma_r, q_r}$ by $\pi$
\State $S\gets S\cup \{v^r_{\pi(f)}\cdot \omega^{-a^T\Sigma f} \}$
\EndFor
\State $\eta\gets \text{median}(S)$ \Comment{Take the median coordinatewise}
\State {\bf If ~}{$|\eta|\leq \nu/2$}{~\bf then continue} \Comment{Continue if the estimated value is too small}
\State $L\gets L\cup \{f\}$
\State $ w_f\gets \eta$
\EndFor
\State \textbf{return} $ w$
\EndProcedure 
\end{algorithmic}
\end{algorithm}

We first give intuition behind the algorithm and the analysis. We define the set $S\subseteq \nsq$ to contain elements $i\in \nsq$ such that $|x_i|^2\geq \e \mu^2$ (i.e. $S$ is the set of {\em head elements} of $x$). As we show later (see section~\ref{sec:analysis}) it is sufficient to locate and estimate all elements in $S$ up to $O(\e \mu^2)$ error term in order obtain $\ell_2/\ell_2$ guarantees that we need\footnote{In fact, one can see that our algorithm gives the stronger $\ell_\infty/\ell_2$ guarantee}. Algorithm~\ref{alg:sfft} performs $O(\log N)$ rounds of location and estimation, where in each round the located elements are estimated up to a constant factor. The crucial fact that allows us to obtain an optimal sampling bound is that the algorithm uses the same samples during these $O(\log N)$ rounds. 
Thus, our main goal  is to show that elements of $S$ will be successfully located and estimated throughout the process, {\em despite the dependencies} between the sampling pattern and the residual signal $x-\chi^{(t)}$ that arise  due to reuse of randomness in the main loop of Algorithm~\ref{alg:sfft}. 

We now give an overview of the main ideas that allow us to circumvent lack of independence. Recall that our algorithm needs to estimate all {\em head elements}, i.e. elements $i\in S$, up to $O(\e \mu^2)$ additive error. Fix an element $i\in S$ and for each permutation $\pi$ consider balls $\B^\infty_{\pi(i)}((n/b)\cdot 2^{t+2})$ around the position that $i$ occupies in the permuted signal. For simplicity, we assume that $d=1$, in which case the balls $\B^\infty_{\pi(i)}((n/b)\cdot 2^{t+2})$ are just intervals:
\begin{equation}\label{eq:ball-1d}
\B^\infty_{\pi(i)}((n/b)\cdot 2^{t+2})=\pi(i)+[-(n/b)\cdot 2^{t+2}, +(n/b)\cdot 2^{t+2}],
\end{equation}
where addition is modulo $n$. Since our filtering scheme is essentially ``hashing'' elements of $x$ into $B=\Omega(|S|/\alpha)$ ``buckets'' for a small constant $\alpha>0$, we expect at most $O(\alpha)2^{t+2}$ elements of $S$ to land in a ball \eqref{eq:ball-1d} (i.e. the expected number of elements that land in this ball is proportional to its volume). 

First suppose that this expected case occurs for any permutation, and assume that all head elements (elements of $S$) have the same magnitude (equal to $1$ to simplify notation).  It is now easy to see that the number of elements of $S$ that are mapped to \eqref{eq:ball-1d} {\em for any $t\geq 0$} does not exceed its expectation (we call element $i$ ``isolated'' with respect to $\pi$ {\em at scale $t$} in that case), then the contribution of $S$  to $i$'s estimation error is $O(\alpha)$.  Indeed, recall that the contribution of an element $j\in \nsq$ to the estimation error of $i$ is about $(1+(b/n)|\pi(i)-\pi(j)|)^{-F}$ by Lemma~\ref{lm:filter-prop}, (2), where we can choose $F$ to be any constant without affecting the asymptotic sample complexity.  Thus, even if $F=2$, corresponding to the boxcar filter, the contribution to $i$'s estimation error is bounded by 
\begin{equation*}
\begin{split}
&\sum_{t\geq 0, (n/b)\cdot 2^{t+2}<n/2} \left|\pi(S)\cap  \B^\infty_{\pi(i)}((n/b)\cdot 2^{t+2})\right|\cdot \max_{y\in \B^\infty_{\pi(i)}((n/b)\cdot 2^{t+2})\setminus \B^\infty_{\pi(i)}((n/b)\cdot 2^{t+1})} |G_{\pi(i)-y}|\\
&=\sum_{t\geq 0, (n/b)\cdot 2^{t+2}<n/2} O(\alpha 2^{t+2})\cdot (1+2^{t+1})^{-F}=O(\alpha).\\
\end{split}
\end{equation*}
Thus, if not too many elements of $S$ land in intervals around $\pi(i)$, then the error in estimating $i$ is at most $O(\alpha)$ times the maximum head element in the current residual signal (plus noise, which can be handled separately). This means that the median in line 15 of Algorithm~\ref{alg:loc-est} is an additive $\pm O(\alpha) ||x-\chi||_\infty$ approximation to element $f$. Since Algorithm~\ref{alg:loc-est} only updates elements that pass the magnitude test in line 16, we can conclude that whenever we update an element, we have a $(1\pm O(\alpha))$ multiplicative estimate of its value, which is sufficient to conclude that we decrease the $\ell_\infty$ norm of $x-\chi$ in each iteration. Finally, we crucially ensure that the signal is never updated outside of the set $S$. This means that the set of head elements is fixed in advance and does not depend on the execution path of the algorithm! This allows us to formulate a notion of isolation with respect to the set $S$ of head elements fixed in advance, and hence avoid issues arising from the lack of independence of the signal $x-\chi$ and the permutaions we choose. 

 We formalize this notion in Definition~\ref{def:isolated}, where we define what it means for $i\in S$ to be isolated under $\pi$.  Note that the definition is essentially the same as asking that the balls in \eqref{eq:ball-1d} do not contain more than the expected number of elements of $S$. However, we need to relax the condition somewhat in order to argue that it is satisfied with good enough probability {\em simultaneously for all $t\geq 0$}. A adverse effect of this relaxation is that our bound on the number of elements of $S$ that are mapped to a balls around $\pi(i)$ are weaker than what one would have in expectation. This, however, is easily countered by choosing a filter with stronger, but still polynomial, decay (i.e. setting the parameter $F$ in the definiion of our filter $G$ in \eqref{eq:def-g} sufficiently large).
 
 As noted before, the definition of being isolated crucially only depends  on the {\bf locations} of heavy hitters as opposed to their {\bf values}. This allows us to avoid an (intractable) union bound over all signals that appear during the execution of our algorithm. We give formal definitions of isolationin section~\ref{sec:isolated}, and then use them to analyze the algorithm in section~\ref{sec:analysis}.

\section{Isolated elements and main technical lemmas}\label{sec:isolated}
 We now give the technical details for the outline above.

\newcommand{\defsk}{Let $S\subseteq \nsq, |S|\leq 2k/\e,$ be such that $||x_{\nsq\setminus S}||_\infty\leq \mu$.  ~Let $B\geq k/(\e \alpha^d)$.~}
\newcommand{\defskmin}{Let $S\subseteq \nsq, |S|\leq 2k/\e$.  ~Let $B\geq k/(\e \alpha^d)$.~}

\begin{definition}
For a permutation $\pi$ and a set $S\subseteq \nsq$ we denote $S^\pi:=\{\pi(x): x\in S\}$.
\end{definition}

\begin{definition}\label{def:isolated}
Let $\Sigma\in \gl, q\in \nsq$, and let $\pi=\pi_{\Sigma, q}$. We say that an element $i$ is {\em isolated} under permutation $\pi$ {\em at scale $t$} if
$$
|(S\setminus \{i\})^\pi\cap \B^\infty_{\pi(i)}((n/b)\cdot 2^{t+2})|\leq \alpha^{d/2} 2^{(t+3)d}\cdot 2^{t}.
$$
We say that $i$ is simply {\em isolated} under permutation $\pi_{\Sigma, q}$ if it is isolated under $\pi_{\Sigma, q}$ at all scales $t\geq 0$.
\end{definition}
\begin{remark}
We will use the definition of isolated elements for a set $S$ with  $|S|\approx k/\e$.
\end{remark}

The following lemma shows that every $i\in \nsq$ is likely to be isolated under a randomly chosen permutation $\pi$:
\begin{lemma}\label{lm:isolated-pi}
\defskmin  Let $\Sigma\in \gl, q\in \nsq$ be chosen uniformly at random, and let $\pi=\pi_{\Sigma, q}$. Then each $i\in \nsq$ is {\em isolated} under permutation $\pi$ with probability at least $1-O(\alpha^{d/2})$.
\end{lemma}
\begin{proof}
  By Lemma~\ref{lemma:limitedindependence}, for any fixed $i$, $j\neq i$ and any radius $r\geq 0$,
\begin{equation}\label{eq:li}
 \prob_{\Sigma}[\norm{\infty}{\Sigma(i-j)} \leq r] \leq 2(2r/n)^d.
\end{equation}
Setting $r=(n/b)\cdot 2^{t+2}$, we get
\begin{equation}\label{eq:lixchbs}
\begin{split}
\expect_{\Sigma, q}[|(S\setminus \{i\})^\pi\cap \B^\infty_{\pi(i)}((n/b)\cdot 2^{t+2})|]&=\sum_{j\in S\setminus \{i\}} \prob_{\Sigma, q}[\pi(j)\in  \B^\infty_{\pi(i)}((n/b)\cdot 2^{t+2})]\\
\end{split}
\end{equation}
Since $\pi_{\Sigma, q}(i)=\Sigma(i-q)$ for all $i\in \nsq$, we have
\begin{equation*}
\begin{split}
\prob_{\Sigma, q}[\pi(j)\in  \B^\infty_{\pi(i)}((n/b)\cdot 2^{t+2})]&=\prob_{\Sigma, q}[||\pi(j)-\pi(i)||_\infty\leq (n/b)\cdot 2^{t+2}]\\
&=\prob_{\Sigma, q}[||\Sigma(j-i)||_\infty\leq (n/b)\cdot 2^{t+2}]\leq 2(2^{t+3}/b)^{d},
\end{split}
\end{equation*}
where we used \eqref{eq:li} in the last step. using this in \eqref{eq:lixchbs}, we get
\begin{equation*}
\begin{split}
\expect_{\Sigma, q}[|(S\setminus \{i\})^\pi\cap \B^\infty_{\pi(i)}((n/b)\cdot 2^{t+2})|]&\leq |S|\cdot (2^{t+3}/b)^{d}\leq (|S|/B)\cdot 2^{(t+3)d}\lesssim \e \alpha^d 2^{(t+3)d}.
\end{split}
\end{equation*}
Now by Markov's inequality we have that $i$ fails to be isolated at scale $t$ with probability at most
$$
\prob_{\Sigma, q}\left[|(S\setminus \{i\})^\pi\cap \B^\infty_{\pi(i)}((n/b)\cdot 2^{t+2})|>\alpha^{d/2} 2^{(t+3)d+t}\right]\lesssim 2^{-t} \alpha^{d/2}.
$$
Taking the union bound over all $t\geq 0$, we get
$$
\prob_{\Sigma, q}[i~\text{is not isolated}]\lesssim \sum_{t\geq 0}2^{-t} \alpha^{d/2} \lesssim \alpha^{d/2}
$$
as required.

\end{proof}

The contribution of tail noise to an element $i\in \nsq$ is captured by the following
\begin{definition}
Let $x\in \C^N$. \defskmin Let $u=\sqrt{N}\F^{-1}((P_{\Sigma, a, q}\hat x)\cdot \hat G)$. We say that an element $i\in \nsq$ is {\em well-hashed with respect to noise under $(\pi_{\Sigma, q}, a)$} if 
$$
|u_{\pi(i)}\omega^{-a^T\Sigma i}-x_i|^2=O(\sqrt{\alpha}) \e\mu^2,
$$
where we let $\pi=\pi_{\Sigma, q}$ to simplify notation.
\end{definition}

\begin{lemma}\label{lm:good-prob}
\defsk Let $\Sigma_r\in \gl, q_r, a_r\in \nsq, r=1,\ldots, r_{max}, r_{max}\geq (C/\sqrt{\alpha})\log N$ be chosen uniformly at random, where $\alpha>0$ is a constant and $C>0$ is a sufficiently large constant that depends on $\alpha$.  Then with probability at least $1-N^{-\Omega(C)}$
\begin{enumerate}
\item each $i\in \nsq$ is isolated with respect to $S$ under at least $(1-O(\sqrt{\alpha}))r_{max}$ permutations $\pi_r, r=1,\ldots, r_{max}$;
\item each $i\in \nsq$ is well-hashed with respect to noise under  at least $(1-O(\sqrt{\alpha}))r_{max}$ pairs  $(\pi_r, a_r), r=1,\ldots, r_{max}$.
\end{enumerate}
\end{lemma}
\begin{proof}
The first claim follows by an application of Chernoff bounds and Lemma~\ref{lm:isolated-pi}.
For the second claim, let $u=\sqrt{N}\F^{-1}((P_{\Sigma, a, q}\hat x_{\nsq\setminus S})\cdot \hat G)$, where $(\Sigma, a, q)=(\Sigma_r, a_r, q_r)$ for some $r=1,\ldots, r_{max}$.
Letting $\pi=\pi_{\Sigma, q}$, by Lemma~\ref{lm:hashing}, (1) and (2) we have 
$$
\expect_{\Sigma, q, a}[\abs{u_{\pi(i)}\omega^{-a^T\Sigma i} - (x_{\nsq\setminus S})_i}^2]\leq (C')^d ||x_{\nsq\setminus S}||^2/B+||x||^2\cdot N^{-\Omega(c)}
$$
for a constant $C'>0$, 
where we asssume that arithmetic operations are performed on $c\log N$-bit numbers for some constant $c>0$. Since we assume that $R^*\leq \poly(N)$, we have
$$
\expect_{\Sigma, q, a}[\abs{u_{\pi(i)}\omega^{-a^T\Sigma i} - (x_{\nsq\setminus S})_i}^2]\leq (C'')^d ||x_{\nsq\setminus S}||^2/B+\mu^2\cdot N^{-\Omega(c)}.
$$
Let $S^*\subset\nsq$ denote a set of top $k$ coefficients of $x$ (with ties broken arbitrarily). We have
\begin{equation*}
\begin{split}
&||x_{\nsq\setminus S}||^2\leq ||x_{\nsq\setminus (S\cup S^*)}||^2+||x_{S^*\setminus S}||^2\leq||x_{\nsq\setminus S}||^2\leq ||x_{\nsq\setminus S^*}||^2+k\cdot ||x_{\nsq\setminus S}||_\infty^2\leq 2k \mu^2.
\end{split}
\end{equation*}
Since $B\geq k/(\e \alpha^d)$, we thus have
$$
\expect_{\Sigma, q, a}[\abs{u_{\pi(i)}\omega^{-a^T\Sigma i} - (x_{\nsq\setminus S})_i}^2]\leq (C'''\alpha)^d \e\mu^2
$$
for a constant $C'''>0$.

By Markov's inequality 
$$
\prob_{\Sigma, q, a}[\abs{u_{h(i)}\omega^{-a^T\Sigma i} - (x_{\nsq\setminus S})_i}^2>(C'''\sqrt{\alpha})^d \e \mu^2]<\alpha^{d/2}.
$$
As before, an application of Chernoff bounds now shows that each $i\in \nsq$  is well-hashed with respect to noise with probability at least $1-N^{-10}$, and hence all $i\in \nsq$ are well-hashed with respect to noise with probability at least $1-N^{-\Omega(C)}$ as long as $\alpha$ is smaller than an absolute constant.
\end{proof}

We now combine Lemma~\ref{lm:isolated-pi} with Lemma~\ref{lm:good-prob} to derive a bound on the noise in the ``bucket'' of an element $i\in \nsq$ due to both heavy hitters and tail noise. Note that crucially, the bound only depends on the $\ell_\infty$ norm of the head elements (i.e. the set $S$), and in particular, works for {\em any signal} that coincides with $x$ on the complement of $S$.  Lemma~\ref{lm:small-noise} will be the main tool in the analysis of our algorithm in the next section.
\begin{lemma}\label{lm:small-noise}
Let $x\in \C^N$.  \defsk Let $y\in \C^N$ be such that 
$y_{\nsq\setminus S}=x_{[n]\setminus S}$ and $||y_{S}||_\infty\leq 4\sqrt{\e}\mu 2^w$ for some $t\geq 0$.  Let $u=\sqrt{N}\F^{-1}((P_{\Sigma, a, q}\hat y)\cdot \hat G)$. Then for each  $i\in \nsq$ that is isolated and well-hashed with respect to noise under $(\Sigma_r, q_r, a_r)$ one has 
$$
|u_j\omega^{-a^T\Sigma i}-y_i|^2\lesssim \sqrt{\alpha} \e((4\mu 2^w)^2+\mu^2).
$$
\end{lemma}
\begin{proof}
We have
\begin{equation*}
\left|u_j \omega^{-a^T\Sigma i}-x_i\right|^2\leq 2 |A^{H}|^2+2|A^T|^2,
\end{equation*}
where $A^H=u^H_j \omega^{-a^T\Sigma i}-(y_{S})_i$
for $u^H=\sqrt{N}\F^{-1}((P_{\Sigma, a, q}\hat y_S)\cdot \hat G)$
and 
$A^T=u^T_j \omega^{-a^T\Sigma i}-(y_{\nsq\setminus S})_i$
for $u^T=\sqrt{N}\F^{-1}((P_{\Sigma, a, q}\hat y_{\nsq \setminus S})\cdot \hat G)$

We first bound $A^H$. Fix $i\in \nsq$. If $i$ is isolated, we have
$$
|(S\setminus \{i\})^\pi\cap \B^\infty_{\pi(i)}((n/b)\cdot 2^{t+2})|\leq \alpha^{d/2} 2^{(t+3)d}\cdot 2^{t}.
$$
for all $t\geq 0$.  
We have
\begin{equation*}
\begin{split}
|A^H|=|u^{H}_j \omega^{-a^T\Sigma i}-y_{i}|&=\left|\sum_{j \in S\setminus \{i\}} y_j G_{o_i(j)} \omega^{-a^T\Sigma j}\right|\\
 &\leq \sum_{j \in S\setminus \{i\}} |y_j G_{o_i(j)}|\\
&\leq ||y_{S}||_\infty\cdot \sum_{t\geq 0} \left(\frac{2}{1+2^{t+2}}\right)^{\fc} |(S\setminus \{i\})^\pi\cap \B^\infty_{\pi(i)}((n/b)\cdot 2^{t+2})|\\
&\leq ||y_{S}||_\infty\cdot \sum_{t\geq 0} \left(\frac{2}{1+2^{t+2}}\right)^{\fc} |(S\setminus \{i\})^\pi\cap \B^\infty_{\pi(i)}((n/b)\cdot 2^{t+2})|\\
& \leq ||y_{S}||_\infty\cdot \sum_{t\geq 0} \left(\frac{2}{1+2^{t+2}}\right)^{\fc} \alpha^{d/2} 2^{(t+3)d}\cdot 2^{t}=||y_{S}||_\infty\cdot  O(\alpha^{d/2})=O(\alpha^{d/2} \sqrt{\e}\mu 2^w)
\end{split}
\end{equation*}
as long as $\fc\geq 2d, \fc=\Theta(d)$.
Further, if $i$ is well-hashed with respect to noise, we have
$|A^T|\lesssim \sqrt{\alpha} \e\mu^2$.
Putting these estimates together yields the result. 
\end{proof}

\section{Main result}\label{sec:analysis}

In this section we use Lemma~\ref{lm:small-noise} to prove that our algorithm satisfies the stated $\ell_2/\ell_2$ sparse recovery guarantees. The proof consists of two main steps: Lemma~\ref{lm:loc-est} proves that one iteration of the peeling process (i.e. one call to \textsc{LocateAndEstimate})  outputs a list containing all elements whose values are close to the current $\ell_\infty$ norm of the residual signal. Furthermore, approximations that \textsc{LocateAndEstimate} returns for elements in its output list are correct up to a multiplicative $1\pm 1/3$ factor. Lemma~\ref{lm:linf-bound} then shows that repeated invocations of \textsc{LocateAndEstimate} reduce the $\ell_\infty$ norm of the residual signal as claimed.

\begin{lemma}\label{lm:loc-est}
Let $x\in \C^N$. \defsk  Consider the $t$-th iteration of the main loop in Algorithm~\ref{alg:sfft}. Suppose that $||x-\chi||_\infty\leq 2\nu$. Suppose that each element $i\in \nsq$ 
is isolated with respect to $S$ and well-hashed with respect to noise  under at least $(1-O(\sqrt{\alpha}))r_{max}$ values of $r=1,\ldots, r_{max}$.  Let $y=x-\chi^{(t)}$, and let $\chi'$ denote the output of \textsc{LocateAndEstimate}. Then 
one has 
\begin{enumerate}
\item $|\chi'_i- y_i|<\frac1{3} |y_i|$ for all $i\in L$;
\item  all $i$ such that $|y_i|\geq \nu$ are included in $L$.
\end{enumerate}
as long as $\alpha>0$ is a sufficiently small constant.
\end{lemma}
\begin{proof}
We let $y:=x-\chi^{(t)}$ to simplify notation.  Fix $i\in \nsq$. 

Consider $r$ such that $i$ is isolated under $\pi_r$ and well-hashed with respect to noise under $(\pi_r, a_r)$. Then we have by Lemma~\ref{lm:small-noise}
\begin{equation}\label{eq:err-small}
\begin{split}
|v^r_j\omega^{-a^T\Sigma i}-y_i|^2& \lesssim \sqrt{\alpha}(\e(||y_{[S]}||_\infty)^2+\e\mu^2)\lesssim \sqrt{\alpha}\e((2\nu)^2+\mu^2)\leq (\frac1{16} \nu)^2\\
\end{split}
\end{equation}
as long as $\alpha$ is smaller than an absolute constant.

Since each $i$ is well-hashed with respect to at least at $1-O(\sqrt{\alpha})$ fraction of permutations, we get that $|y_i-\eta|\leq \frac1{16}\nu$.
Now if $i\in L$, it must be that $|\eta|>\nu/2$, but then 
\begin{equation}\label{eq:sfokgj}
|y_i|\geq |\eta|-\nu/16>\nu/2-\nu/16>(3/4)\nu.
\end{equation}
This also implies that $|\chi'_i-y_i|\leq \nu/16<(4/3)|y_i|/16<|y_i|/3$, so the first claim follows.

For the second claim, it suffices to note that if $|y_i|>\nu$, then we must have $|\eta|\geq |y_i|-\nu/16>\nu/2$, so $i$ passes the magnitude test and is hence included in $L$.
\end{proof}

We can now prove the main lemma required for analysis of Algorithm~\ref{alg:sfft}:
\begin{lemma}\label{lm:linf-bound}
Let $x\in \C^N$. Let $\Sigma_r\in \gl, a_r, q_r\in \nsq, r=1,\ldots, r_{max}=C\log N$, where $C>0$ is a sufficiently large constant, be chosen uniformly at random. Then with probability at least $1-N^{-\Omega(C)}$ one has $||x-\chi^{(T-1)}||_\infty\leq 4\sqrt{\e}\mu$. 
\end{lemma}
\begin{proof}
We now fix a specific choice of the set $S\subseteq \nsq$. Let
\begin{equation}\label{eq:s-def}
S=\{i\in \nsq: |x_i|>\sqrt{\e}\mu\}.
\end{equation}

First note that $||x_{\nsq\setminus S}||_\infty\leq \mu$. Also, we have $|S|\leq 2k/\e$. Indeed, recall that $\mu^2=\err_k^2(x)/k$. If $|S|>2k/\e$, more than $k/\e$ elements of $S$ belong to the tail, amounting to at least $\e\mu^2\cdot (k/\e)>\err_k^2(x)$ tail mass.  Thus, since $r_{max}\geq C\log N$, and by the choice of $B$ in Algorithm~\ref{alg:sfft}, we have by Lemma~\ref{lm:good-prob} that with probability at least $1-N^{-\Omega(C)}$
\begin{enumerate}
\item each $i\in \nsq$ is isolated with respect to $S$ under at least $(1-O(\sqrt{\alpha}))r_{max}$ permutations $\pi_r, r=1,\ldots, r_{max}$;
\item each $i\in \nsq$ is well-hashed with respect to noise under  at least $(1-O(\sqrt{\alpha}))r_{max}$ permutations $\pi_r, r=1,\ldots, r_{max}$.
\end{enumerate}
This ensures that the preconditions of Lemma~\ref{lm:loc-est} are satisfied.
We now prove the following statement for $t\in [0:T]$ by induction on $t$:
\begin{enumerate}
\item $\chi^{(t)}_{[n]\setminus S}\equiv 0$
\item $||(x-\chi^{(t)})_{S}||_\infty\leq 4\sqrt{\e}\mu 2^{T-t}$.
\item $|x_i-\chi^{(t)}_i|\leq |x_i|$ for all $i\in \nsq$.
\end{enumerate}

\begin{description}
\item[Base:$t=0$] True by the choice of $T$.
\item[Inductive step: $t\to t+1$] Consider the list $L$ constructed by \textsc{LocateAndEstimate} at iteration $t$. 
Let $S^*:=\{i\in S: |(x-\chi^{(t)})_i|>4\sqrt{\e}\mu 2^{T-(t+1)}\}$.  We have $S^*\subseteq L$ by Lemma~\ref{lm:loc-est}, (2).
Thus, 
$$
||(x-\chi^{(t+1)})_{S}||_{\infty}\leq \text{max}\{||(x-\chi^{(t+1)})_{S^*}||_{\infty}, ||(x-\chi^{(t+1)})_{S\setminus S^*}||_{\infty}, ||(x-\chi^{(t+1)})_{\nsq\setminus S}||_{\infty}\}.
$$
By Lemma~\ref{lm:loc-est}, (1) we have $|x_i-\chi^{(t+1)}_i|\leq |x_i-\chi^{(t)}_i|/3$ for all $i\in L$, so (3) follows. Furthemore, this implies that
\begin{enumerate}
\item $||(x-\chi^{(t+1)})_{S^*}||_{\infty}\leq ||x-\chi^{(t)}||_{\infty}/3\leq 4\sqrt{\e}\mu 2^{T-(t+1)}$ by the inductive hypothesis;
\item $||(x-\chi^{(t+1)})_{S\setminus S^*}||_{\infty}\leq 4\sqrt{\e}\mu 2^{T-(t+1)}$ by definition of $S^*$;
\item $||(x-\chi^{(t+1)})_{\nsq \setminus S}||_{\infty}=||x_{\nsq \setminus S}||_{\infty}\leq 4\sqrt{\e}\mu 2^{T-(t+1)}$ by the inductive hypothesis together with the definition of $\mu$ and the fact that $t\leq T-1$.
\end{enumerate}
This proves (2).

Finally, by Lemma~\ref{lm:loc-est}, (2) only elements $i$ such that $|(x-\chi^{(t)})_i|>\frac{3}{4} 4\sqrt{\e}\mu2^{T-t}\geq (3/2)4\sqrt{\e}\mu$ are included in $L$. Since $|(x-\chi^{(t)})_i|\leq |x_i|$, this means that $|x_i|>4\mu$, i.e. $i\in S$ and $\chi^{(t+1)}_{[n]\setminus S}=0$, as required.
\end{description}
\end{proof}

We can now prove 
\begin{theorem}
Algorithm~\ref{alg:sfft} returns a vector $\chi$  such that 
$$
||x-\chi||_2\leq (1+O(\e))\err_k(x).
$$
The number of samples is bounded by $d^{O(d)}\frac1{\e}k\log N$, and the runtime is bounded by $O(N\log^3 N)$.
\end{theorem}
\begin{proof}
By Lemma~\ref{lm:linf-bound} we have by setting $t=T-1$
$||x-\chi||_\infty\leq 8\sqrt{\e}\mu$,
so
$$
||x-\chi||^2_2\leq ||(x-\chi)_{[k]}||^2_\infty\cdot k+||(x-\chi)_{\nsq\setminus [k]}||^2_2\leq ||(x-\chi)_{[k]}||^2_\infty\cdot k+||x_{\nsq\setminus [k]}||^2_2\leq (1+O(\e)) \err^2_k(x),
$$
where we used Lemma~\ref{lm:linf-bound}, (3) to upper bound 
$||(x-\chi)_{\nsq\setminus [k]}||^2_2$ with $||x_{\nsq\setminus [k]}||^2_2$.

We now bound sampling complexity. The support of the filter $G$ is bounded by $O(B \fc^d)=d^{O(d)}\frac1{\e} k$ by construction. We are using $r_{max}=\Theta(\log N)$, amounting to  $d^{O(d)}\frac1{\e} k \log N$ sampling complexity overall. The location and estimation loop takes $O(N\log^3 N)$ time: each time the vector $v^r$ is calculated in \textsc{LocateAndEstimate} $O(N\log N)$ time is used by the FFT computation,  so since we compute $O(\log N)$ vectors during each of $O(\log N)$ iterations, this results in an $O(N\log^3 N)$ contribution to runtime.
\end{proof}

%% file: experiments.tex
\section{Experimental evaluation}\label{sec:experiments}
In this section we describe results of an experimental evaluation of our algorithm from section~\ref{sec:algo-linear}. 
In order to avoid the issue of numerical precision and make the notion of recovery probability well-defined, we focus  the problem of {\em support recovery}, where the goal is to recover the {\em positions} of the non-zero coefficients.
We first describe the experimental setup that we used to evaluate our algorithm, and follow with evaluation results.

\subsection{Experimental setup}
We present experiments for support  recovery from one-dimensional Fourier measurements (i.e. $d=1$). In this problem one is given frequency domain access to a signal $\wh{x}$ that is {\em exactly} $k$-sparse  in the time domain, and needs to recovery the support of $x$ exactly. 
The support of $x$ was always chosen to be uniformly among subsets of $[N]$ of size $k$. We denote the sparsity by $k$, and the support of $x$ by $S\subseteq \nsq$.

We compared our algorithm to two algorithms for sparse recovery:
\begin{itemize}
\item $\ell_1$-minimization, a state-of-the-art technique for practical sparse recovery using Gaussian and Fourier measurements. The best known sample bounds for the sample complexity in the case of approximate sparse recovery are $O(k\log^3 k\log N)$~\cite{CTao,RV, CGV}. The running time of $\ell_1$ minimization is dominated by solving a linear program. We used the implementation from \textsc{SPGL1}~\cite{BergFriedlander:2008, spgl1:2007}, a standard Matlab package for sparse recovery using $\ell_1$-minimization.
For this experiment we let $x_i$ be chosen uniformly random on the unit circle in the complex plane when $i\in S$ and equal to $0$ otherwise.  

\item Sequential Sparse Matching Pursuit (SSMP)~\cite{BI2009}. SSMP is an iterative algorithm for sparse recovery using sparse matrices. SSMP has optimal $O(k\log (N/k))$ sample complexity bounds and $\tilde O(N)$ runtime. The sample complexity and runtime bounds are similar to that of our algorithm, which makes SSMP a natural point of comparison. Note, however, that the measurement matrices used by SSMP are binary and sparse, i.e.,  very different from the Fourier matrix.
For this experiment we let $x_i$ be uniformly random in $\{-1, +1\}$  when $i\in S$ and $0$ otherwise.  
\end{itemize}

Our implementation of Algorithm~\ref{alg:sfft} uses the following parameters. First, the filter $G$ was the simple boxcar filter with support $B=k+1$. The number of measurements $r_{max}$ was varied between $5$ and $25$, with the phase transition occuring around $r_{max}=18$ for most values of $k$.  The geometric sequence of thresholds that \textsc{LocateAndEstimate} is called with was chosen to be powers of $1.2$ (empirically, ratios closer to $1$ improve the performance of the algorithm, at the expense of increased runtime).  We use $N=2^{15}$ and $k=10, 20,\ldots, 100$ for all experiments. We report empirical probability of recovery estimated from $50$ trials.

In order to solve the support recovery problem using \textsc{SPGL1} and Algorithm~\ref{alg:sfft},  we first let both algorithms recover an approximation $x'$ to $x$, and then let 
$$
S:=\{t\in [N]:|x'_t|\geq 1/2\}
$$
denote the recovered support.

\subsection{Results}

\paragraph{Comparison to $\ell_1$-minimization.}
A plot of recovery probability as a function of the number of (complex) measurements and sparsity for \textsc{SPGL1} and Algorithm~\ref{alg:sfft} is given in Fig.~\ref{fig:us-vs-spgl1}.

\begin{figure}[H]
 \includegraphics[width=3.1in]{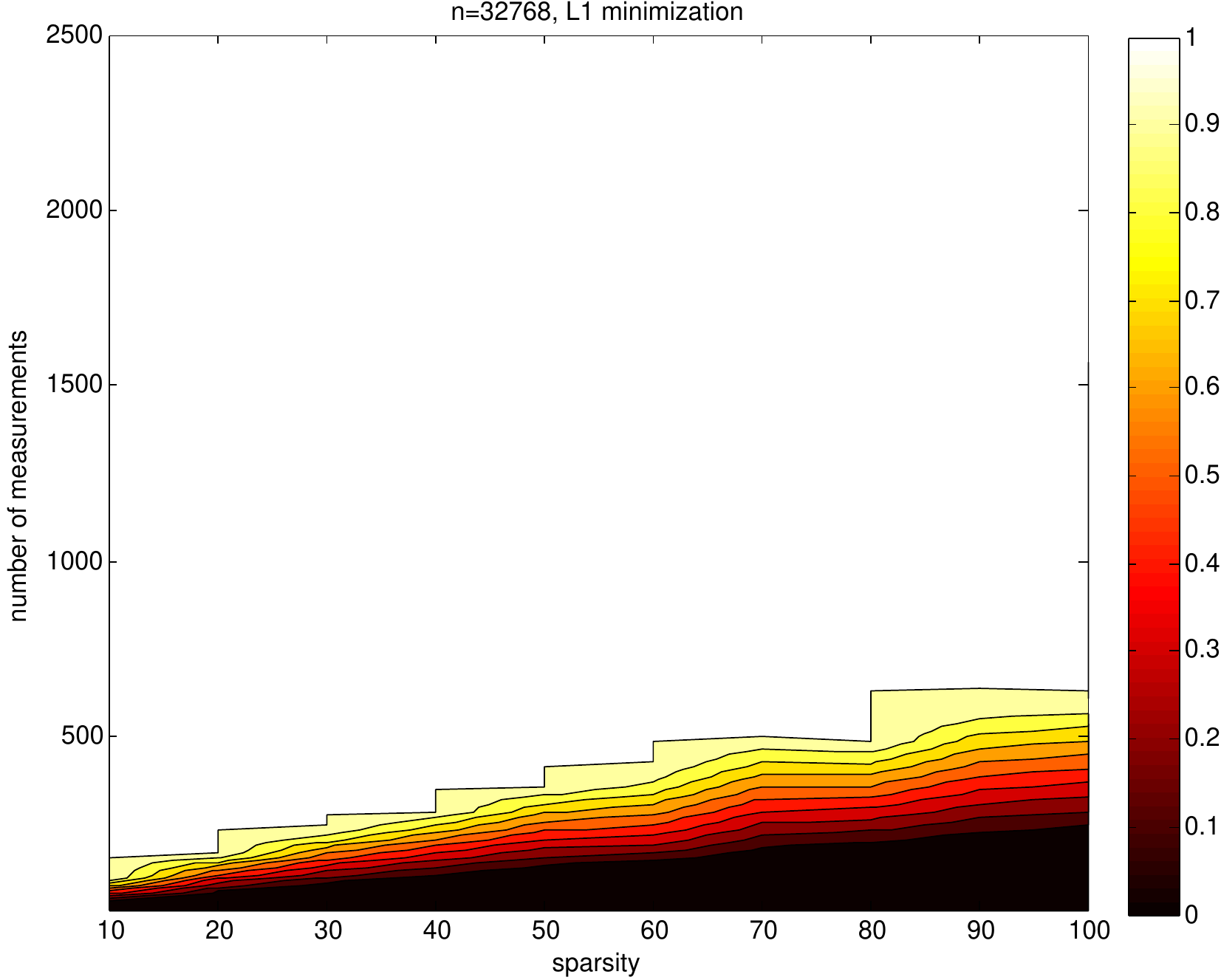}
  \includegraphics[width=3.1in]{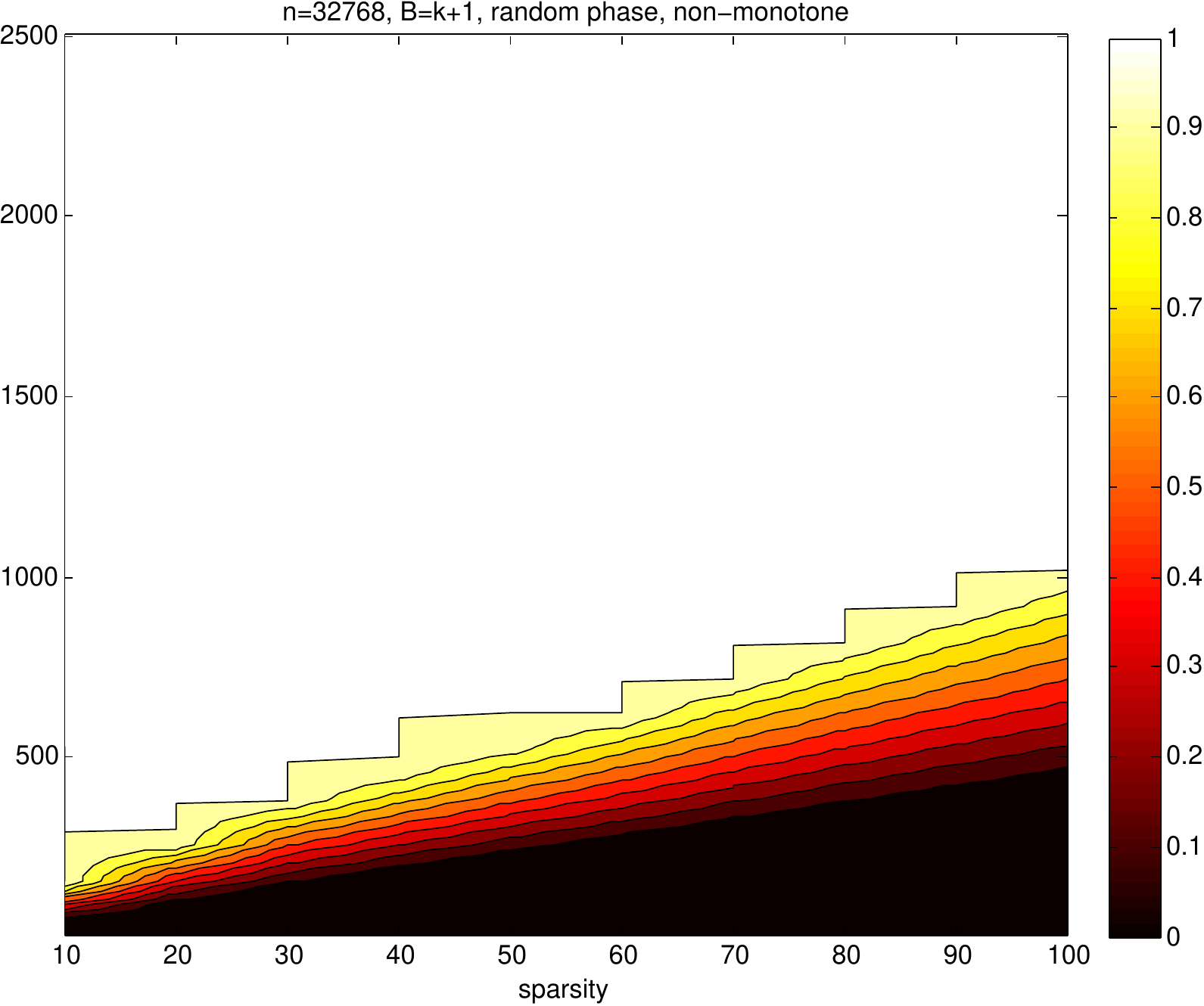}
\caption{Success probability as a function of sparsity and number of measurements: \textsc{SPGL1} (left panel) and Algorithm~\ref{alg:sfft} (right panel). The number of complex measurements is reported.}\label{fig:us-vs-spgl1}
\end{figure}

The empirical sample complexity of our algorithm is within a factor of $2$ of $\ell_1$ minimization if success probability $0.9$ is desired. 
 The best known theoretical bounds for the general setting of approximate sparse recovery show that $O(k\log^3 k \log N)$ samples are sufficient. The runtime is bounded by the cost of solving an $N \times m$ linear program. Our algorithm provides comparable empirical performance, while providing optimal measurement bound and $\tilde O(N)$ runtime.

\paragraph{Comparison to SSMP.}

We now present a comparison to SSMP~\cite{BI2009}, which is a state-of-the art iterative algorithm for sparse recovery using sparse matrices.  We compare our results to experiments in~\cite{Berinde-thesis}.
Since the lengths of signal used for experiments with SSMP in ~\cite{Berinde-thesis} are not powers of $2$, we compare the results of \cite{Berinde-thesis} for $N=20000$ with our results for a larger value of $N$. In particular, we choose $N=2^{15}>20000$. Our results are presented in Fig.~\ref{fig:ssmp}.  Since experiments in ~\cite{Berinde-thesis} used real measurements, we multiply the number of our (complex) measurements by $2$ for this comparison (note that the right panel of Fig.~\ref{fig:ssmp} is the same as the right panel of Fig.~\ref{fig:us-vs-spgl1}, up to the factor of $2$ in the number of measurements). We observe that our algorithm improves upon SSMP by a factor of about $1.15$ when $0.9$ success probability is desired. 

\begin{figure}[H]
 \includegraphics[width=3.1in]{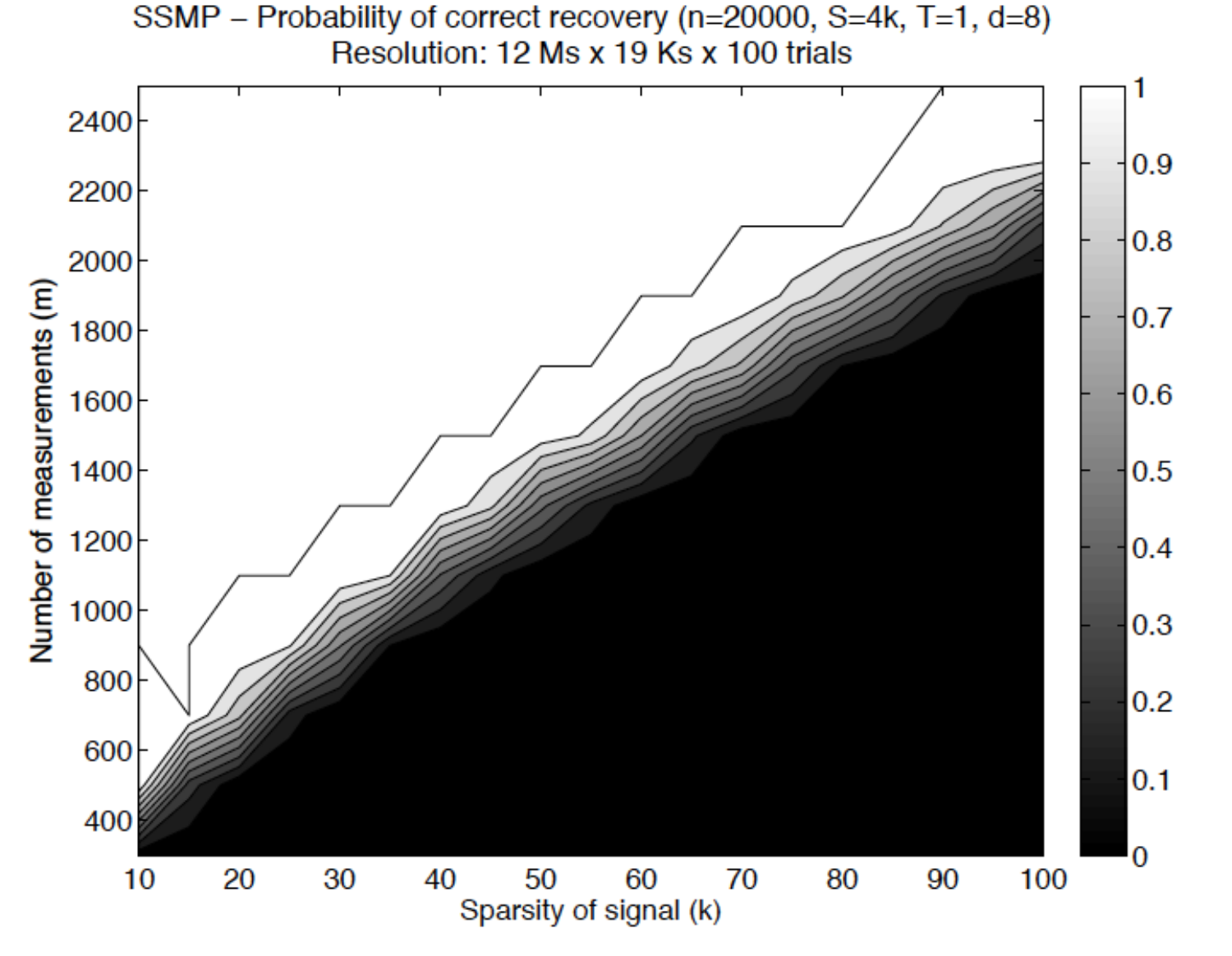}
  \includegraphics[width=3.1in]{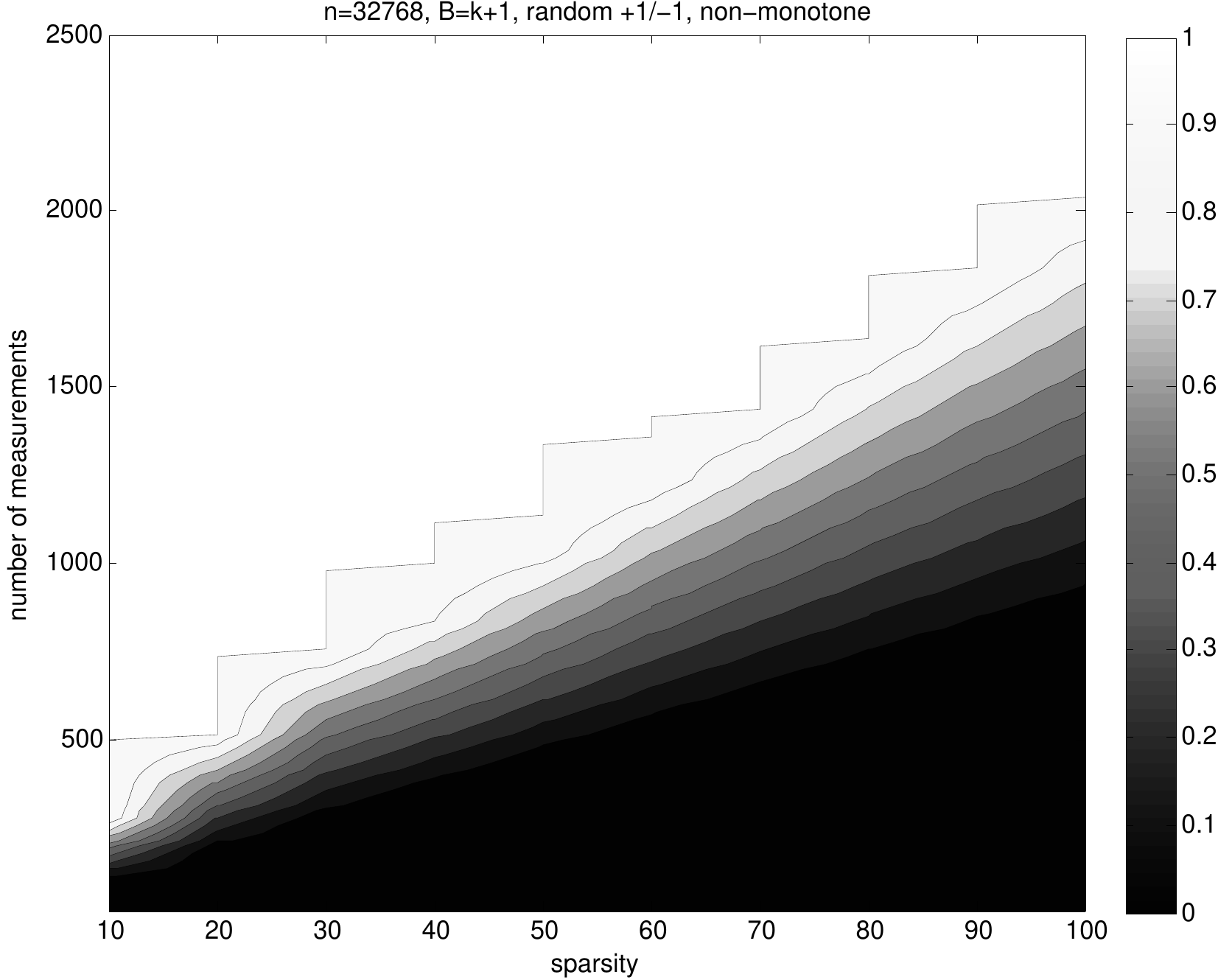}
\caption{Success probability as a function of sparsity and number of  measurements: SSMP (left panel) and Algorithm~\ref{alg:sfft} (right panel). The number of real measurements is reported.}\label{fig:ssmp}
\end{figure}

%% file: app.tex
\section{Omitted proofs}\label{app:A}
\begin{proofof}{Lemma~\ref{lm:perm}}
\begin{equation*}
\begin{split}
\F^{-1}({P_{\Sigma, a, q} \hat x})_{\pi_{\Sigma, q}(i)}=&\frac1{\sqrt{N}}\sum_{j\in \nsq}  \omega^{j^T\Sigma (i-q)}(P_{\Sigma, a, q} \hat x)_j\\
=&\frac1{\sqrt{N}}\sum_{j\in \nsq}  \omega^{j^T\Sigma (i-q)}\hat x_{\Sigma^T(j-a)} \omega^{j^T\Sigma q}\\
=&\frac1{\sqrt{N}}\sum_{j\in \nsq}  \omega^{j^T\Sigma i}\hat x_{\Sigma^T(j-a)}\\
=&\frac1{\sqrt{N}}\sum_{j\in \nsq}  \omega^{i^T\Sigma^T(j-a+a)}\hat x_{\Sigma^T(j-a)}\\
=&\omega^{a^T\Sigma i}\frac1{\sqrt{N}}\sum_{j\in \nsq}  \omega^{i^T\Sigma^T(j-a)}\hat x_{\Sigma^T(j-a)}=\omega^{a^T\Sigma i} x_{i}\\
\end{split}
\end{equation*}
\end{proofof}

\begin{proofof}{Lemma~\ref{lm:filter-prop}}
Since $|\sin (\pi x)|\leq 1, |\sin (\pi x)|\leq |\pi x|$ for all $x$ and $|\sin (\pi x)|\geq 2|x|$ for $|x|\leq 1/2$, we have
\begin{equation}\label{eq:filter-ub}
\left|\frac{\sin(\pi (b-1) j/n)}{(b-1)\sin (\pi j/n)}\right|^{\fc}\leq \left(\frac{2}{(b-1) \pi (j/n)}\right)^{\fc}
\end{equation}
for all $j$. We also have that the maximum absolute value is achieved at $0$. Also, for any $j\in [n]$ such that $|j|\leq \frac{n}{2b}$ one has 
$$
\left|\frac{\sin(\pi (b-1) j/n)}{(b-1)\sin (\pi j/n)}\right|^{\fc}\geq \left|\frac{(1/2)(b-1) j/n}{(b-1) \pi (j/n)}\right|^{\fc}\geq \left|\frac{(1/2)}{\pi}\right|^{\fc}=\frac1{(2\pi)^{\fc}} G_0.
$$
which gives {\bf (1)}.

 {\bf (2)} follows from \eqref{eq:filter-ub} by writing 
 $$
|G_j|=\prod_{s=1}^d H^{\fc}_{j_s} \leq H^{\fc}_{||j||_\infty}\leq \left(\frac{2}{(b-1) \pi (||j||_\infty/n)}\right)^{\fc}\leq \left(\frac{2}{1+(b/n) ||j||_\infty}\right)^{\fc}
 $$
 as long as $b\geq 3$.

\end{proofof}

\begin{proofof}{Lemma~\ref{lemma:limitedindependence}}
We assume wlog that $j=0$. Let $g$ be the largest integer such that $2^g$ divides all of $i_1,\ldots, i_d$. We first assume that $g=0$, and handle the case of general $g$ later. For each $q\geq 0$ let 
\begin{equation}\label{eq:jq-def}
J_q=\{s\in [d]: i_s=2^q e_s, \text{~$e_s$ odd}\}.
\end{equation}
Since we assume that $g=0$, we have $J_0\neq \emptyset$.
We first prove that
  \[
\prob[||Mi||_\infty \leq t] \leq 2(t/n)^d
  \]
 when $M$ is sampled uniformly at random from a set $\mathcal{D}'$ that is a superset of the set of matrices with odd determinant, and then show that $M$ is likely to have odd determinant when drawn from this distribution, which implies the result.
 
We denote the $s$-th column of $M$ by $M_s$.  With this notation we have  $M i=\sum_{s=1}^d M_s i_s$, where $i_s$ is the $s$-th entry of $i$.
Let 
 $$
 \mathcal{D}'=\{M\in \Z^{d\times d}/n: \sum_{s\in J_0} M_{s}\neq {\bf 0}\text{~mod~}2\},
 $$
 i.e. we consider the set of matrices $M$ whose columns with indices in $J_0$ do not add up to the all zeros vector modulo $2$, and are otherwise unconstrained. 
We first note that for any $s\in [1:d]$ we can write
$$
M_s=M'_s+2M_s'',
$$
where $M'_s\in \{0, 1\}^s$ is uniform, and $M_s''$ is uniform in $[n/2]^d$ (and independent of $M_s'$). When $M$ is sampled from $\mathcal{D}'$, one has that $M'_s, s\in J_0$ are conditioned on not adding up to $0$ modulo $2$.

We first derive a more convenient expression for the distribution of $M_s$. In particular, note that $2M_s''$ is distributed identically to $2U$, where $U$ is uniform in $\nsq$ as opposed to $[n/2]^d$. Thus, from now on we assume that we have 
$$
M_s=M'_s+2M_s'',
$$
where $M'_s\in \{0, 1\}$ is uniform and $M_s''$ is uniform in $\nsq$.  One then has for all $s\in J_0$ (all operations are modulo $n$)
$$
M_s i_s=(M'_s+2M_s'')(1+2e_s)=M'_s+2(e_s M'_s+M_s''(1+2e_s))=M'_s+2Q_s,
$$
where $Q_s$ is uniform in $\nsq$ and independent of $M'_s$. This is because $M''_s(1+2 e_s)$ is uniform in $\nsq$ since $M''_s$ is (where we use the fact that $1+2e_s$ is odd). We have thus shown that the distribution of $\sum_{s\in J_0} M_s i_s$ can be generated as follows: one samples bits $M'_s\in \{0, 1\}^d,s\in J_0$ uniformly at random conditional on $\sum_{s\in J_0} M'_s\neq 0$, and then samples $Q_s\in \nsq$ independently and uniformly at random, and outputs 
$$
\sum_{s\in J_0} M'_s+2\sum_{s\in J_0}Q_s.
$$

It remains to note that the distribution of
 $$
\sum_{s\in J_0} M'_s+2\sum_{s\in J_0}Q_s+\sum_{k>0}  \sum_{s\in J_k} M_s i_s
 $$
 is the same as the distribution of 
 $$
\sum_{s\in J_0} M'_s+2\sum_{s\in J_0}Q_s.
 $$
 Indeed, this is because by definition of $\mathcal{D}'$ for any $k>0$ one has $i_s=2^k e_s$, and $\{M_s\}_{s\not \in J_0}$ are independent uniform in $\nsq$. As a consequence,
\begin{equation*}
\begin{split}
 2\sum_{s\in J_0}Q_s+\sum_{k>0}  \sum_{s\in J_k} M_s i_s&= 2\sum_{s\in J_0}Q_s+\sum_{k>0}  2^k\sum_{s\in J_k} M_s e_s\\
& = 2(\sum_{s\in J_0}Q_s+\sum_{k>0}  2^{k-1}\sum_{s\in J_k} M_s e_s)\\
\end{split}
\end{equation*}
 is distributed as $2U$, where $U$ is uniform in $\nsq$. 
 
 Thus, when $M$ is drawn uniformly from $\mathcal{D}'$,  we have that $\sum_{s=0}^d M_s i_s=Mi$ is uniformly random in $\nsq\setminus 2\nsq$, so
\begin{equation*}
\prob_{\Sigma\sim UNIF(\mathcal{D}')}[||\Sigma(i-j)||_\infty \leq t] \leq \frac1{1-2^{-d}}(t/n)^d\leq 2(t/n)^d.
\end{equation*}
  
So far we assumed that $g=0$. However, in general we can divide $i$ by $2^g$, concluding that $M(i/2^g)$ is uniform in $\nsq\setminus 2\nsq$, i.e. $Mi$ is uniform in $2^g\cdot (\nsq\setminus 2\nsq)$, and the same conclusion holds.

  It remains to deduce the same property when $\Sigma$ is drawn uniformly at random from $\gl$, the set of matrices over $\Z^{d\times d}$ with odd determinant. Denote the set of such matrices by $\mathcal{D}_{*}$. First note that 
  $$
  \mathcal{D}_* \subset \mathcal{D}'
  $$
  since the columns $J_0$ of a matrix $\Sigma\in \mathcal{D}_*$ cannot add up to $0$ mod $2$ since any such matrix would not be invertible over $\Z_n$.
We now use the fact that a nonzero polynomial of degree at most $d$ over $\Z_2$ is equal to $1$ with probability at least $2^{-d}$ under a uniformly random assignment, as follows, for example, from the fact that the minimum distance of a $d$-th order Reed-Muller code over $m$ binary variables is $2^{m-d}$(see, e.g. \cite{bok:MW}), 
  we have
  $$
  \prob_{\Sigma\sim UNIF(\Z^{d\times d})}[\Sigma\in \mathcal{D}_*]\geq 2^{-d},
  $$
  and hence
   $$
  \prob_{\Sigma\sim \mathcal{D}'}[\Sigma\in \mathcal{D}_*]\geq 2^{-d}.
  $$
We now get
    \[
\Pr_{\Sigma\sim UNIF(\mathcal{D}_{*})}[||\Sigma(i-j)||_\infty \leq t] \leq \Pr_{\Sigma\sim UNIF(\mathcal{D}')}[||\Sigma(i-j)||_\infty \leq t]/\Pr_{\Sigma\sim UNIF(\mathcal{D}')}[\Sigma \in \mathcal{D}_{*}] \leq 2(2t/n)^d.
  \]
  as required.
\end{proofof}

\begin{proofof}{Lemma~\ref{lm:hashing}}
  By Lemma~\ref{lemma:limitedindependence}, for any fixed $i$ and $j$ and any $t\geq 0$,
  \[
 \prob_{\Sigma}[\norm{\infty}{\Sigma(i-j)} \leq t] \leq 2(2t/n)^d.
  \]

We have
\begin{equation}\label{eq:u-delta}
  u_{\pi(i)} = \sum_{j\in \nsq} G_{o_i(j)}x_j \omega^{a^T \Sigma j}+\Delta_{\pi(i)}
\end{equation}
  for some $\Delta$ with $\norm{\infty}{\Delta} \leq ||x||_1\cdot N^{-\Omega(c)}\leq ||x||_2\cdot N^{-\Omega(c)}$ since we are assuming that arithmetic is performed using $c\log N$ bit words for a constant $c>0$ that can be chosen sufficiently large.
  We define the vector $v \in \C^n$ by $v_{\Sigma j} = x_j G_{o_i(j)}$, so that
  \[
  u_{\pi(i)} - \Delta_{\pi(i)} = \sum_{j\in \nsq} \omega^{a^Tj} v_j = \sqrt{N}\wh{v}_a
  \]
  so
  \[
  u_{\pi(i)} - \omega^{a^T\Sigma i}x_i - \Delta_{\pi(i)} = \sqrt{N}(\wh{v_{\overline{\{\Sigma i\}}}})_a.
  \]
  
  We have by  \eqref{eq:u-delta} and the fact that $(X+Y)^2\leq 2X^2+2Y^2$ 
  \begin{equation*}
  \begin{split}
   \abs{u_{\pi(i)}\omega^{-a^T\Sigma i} - x_i}^2 =\abs{u_{\pi(i)} - \omega^{a^T\Sigma i}x_i}^2\\
    \leq 2\abs{u_{\pi(i)} - \omega^{a^T\Sigma i}x_i - \Delta_{\pi(i)}}^2 + 2\Delta_{\pi(i)}^2\\    
    =2\abs{\sum_{j\in \nsq} G_{o_i(j)}x_j \omega^{a^T \Sigma j}}^2 +2\Delta_{\pi(i)}^2\\        
  \end{split}
  \end{equation*}
  
  By Parseval's theorem, therefore, we have
  \begin{equation}\label{eq:a-est}
  \begin{split}
    \expect_a[\abs{u_{\pi(i)}\omega^{-a^T\Sigma i} - x_i}^2]
    &\leq 2 \expect_a[\abs{\sum_{j\in \nsq} G_{o_i(j)}x_j \omega^{a^T \Sigma j}}^2] + 2\expect_a[\Delta_{h(i)}^2]\\
    &= 2(\norm{2}{v_{\overline{\{\Sigma i\}}}}^2 + \Delta_{h(i)}^2)\\
    &\lesssim\sum_{j\in \nsq \setminus \{i\}} \abs{x_j G_{o_i(j)}}^2+||x||^2_2\cdot N^{-\Omega(c)}\\
    &\lesssim\sum_{j \in \nsq \setminus \{i\}} \abs{x_j G_{o_i(j)}}^2+||x||^2_2\cdot N^{-\Omega(c)}\\
    &\lesssim\mu_{\Sigma, q}^2(i)+||x||_2^2\cdot N^{-\Omega(c)}.\\
    \end{split}
  \end{equation}

We now prove {\bf (2)}. We have
\begin{equation*}
\begin{split}
\expect_{\Sigma, q}[\mu^2_{\Sigma, q}(i)]&=\expect_{\Sigma, q}[\sum_{j \in \nsq\setminus \{i\}} \abs{x_j G_{o_i(j)}}^2].
\end{split}
\end{equation*}
Recall that the filter $G$ approximates an ideal filter, which would be $1$ inside $\B^\infty_{\pi(i)}(n/b)$ and $0$ everywhere else. We use the bound on $G_{o_i(j)}=G_{\pi(i)-\pi(j)}$ in terms of $||\pi(i)-\pi(j)||_\infty$ from Lemma~\ref{lm:filter-prop}, (2). In order to leverage the bound, we partition $\nsq=\B^\infty_{\pi(i)}(n/2)$ as 
$$
\B^\infty_{\pi(i)}(n/2)=\B^\infty_{\pi(i)}(n/b)\cup \bigcup_{t=1}^{\log_2 (b/2)}  \left(\B^\infty_{\pi(i)}((n/b)2^{t})\setminus \B^\infty_{\pi(i)}((n/b)2^{t-1})\right).
$$
For simplicity of notation, let $X_0=\B^\infty_{\pi(i)}(n/b)$ and $X_t=\B^\infty_{\pi(i)}((n/b)\cdot 2^{t})\setminus \B^\infty_{\pi(i)}((n/b)\cdot 2^{t-1})$ for $t\geq 1$.
For each $t\geq 1$ we have by Lemma~\ref{lm:filter-prop}, (2) 
$$
\max_{\pi(l)\in X_t} |G_{o_i(l)}|\leq \max_{\pi(l)\not \in  \B^\infty_{\pi(i)}((n/b)2^{t-1})} |G_{o_i(l)}| \leq \left(\frac{2}{1+2^{t-1}}\right)^{\fc}.
$$
Since the rhs is greater than $1$ for $t\leq 0$, we can use this bound for all $t\leq \log_2 (b/2)$.
Further, by Lemma~\ref{lemma:limitedindependence} we have for each $j\neq i$  and $t\geq 0$
$$
\prob_{\Sigma, q}[\pi(j)\in X_t]\leq \prob_{\Sigma, q}[\pi(j)\in \B^\infty_{\pi(i)}((n/b)\cdot 2^{t})]\leq 2(2^{t+1}/b)^{d}.
$$

Putting these bounds together, we get
\begin{equation*}
\begin{split}
\expect_{\Sigma, q}[\mu^2_{\Sigma, q}(i)]&=\expect_{\Sigma, q}[\sum_{j \in \nsq\setminus \{i\}} \abs{x_j G_{o_i(j)}}^2] \\
&\leq \sum_{j\in \nsq\setminus \{i\}}  \abs{x_j}^2 \cdot \sum_{t=0}^{\log_2 (b/2)} \prob_{\Sigma, q}[\pi(j)\in X_t]\cdot \max_{\pi(l)\in X_t} |G_{o_i(l)}|\\
&\leq \sum_{j\in \nsq\setminus \{i\}}  \abs{x_j}^2 \cdot \sum_{t=0}^{\log_2 (b/2)} (2^{t+1}/b)^{d}\cdot \left(\frac{2}{1+2^{t-1}}\right)^{\fc}\\
&\leq \frac{2^\fc}{B}\sum_{j\in \nsq\setminus \{i\}} \abs{x_j}^2 \sum_{t=0}^{+\infty}  2^{(t+1)d-\fc (t-1)}\\
& \leq 2^{O(d)} \frac{\norm{2}{x}^2}{B}
\end{split}
\end{equation*}
as long as $\fc\geq 2d$ and $\fc=\Theta(d)$, proving {\bf (2)}.
\end{proofof}